\documentclass[a4paper,fleqn]{article}
\usepackage{a4wide}
\usepackage{amsmath,amssymb,amsthm}
\usepackage{thmtools}
\usepackage{enumerate}
\usepackage{tikz}
\usetikzlibrary{arrows,calc,math}
\usetikzlibrary{decorations.pathreplacing}
\usepackage{abstract}
\usepackage{authblk}

\usepackage{hyperref}
\usepackage[T1]{fontenc}
\usepackage{mathpazo}
\makeatletter\@ifpackageloaded{mathpazo}\@tempswatrue\@tempswafalse
\if@tempswa
  \DeclareFontFamily{OT1}{pzc}{}
  \DeclareFontShape{OT1}{pzc}{m}{it}{<-> s * [1.15] pzcmi7t}{}
  \DeclareMathAlphabet{\mathpzc}{OT1}{pzc}{m}{it}
\fi\makeatother
\usepackage{bm}
\usepackage[utf8]{inputenc}

\usepackage[sorting=nyt,giveninits=true,mincrossrefs=4,backend=biber,style=numeric-comp]{biblatex}
\makeatletter\@ifpackageloaded{biblatex}{%
  \usepackage{csquotes} 
  \bibliography{../../references}
  \renewbibmacro{in:}{%
    \ifentrytype{incollection}{\printtext{\bibstring{in}\intitlepunct}}{}}
  \renewbibmacro{publisher+location+date}{%
    \iflistundef{publisher}
      {\setunit*{\addcomma\space}}
      {\setunit*{\addcomma\space}}%
    \printlist{publisher}%
    \setunit*{\addcomma\space}%
    \printlist{location}%
    \setunit*{\addcomma\space}%
    \usebibmacro{date}%
    \newunit}
  \DeclareFieldFormat[article]{pages}{#1\isdot}
  \DeclareFieldFormat[article,incollection,inproceedings,unpublished]{title}{#1\isdot}
  \DeclareFieldFormat[thesis]{title}{\mkbibemph{#1\isdot}}
  \DeclareFieldFormat[unpublished]{date}{(#1)\isdot}
  \DeclareFieldFormat[unpublished]{note}{#1\nopunct} 
  \DeclareFieldFormat[article]{journaltitle}{\mkbibemph{#1\isdot}}
  
  \AtEveryBibitem{%
    \ifentrytype{book}{}{
      \clearname{editor}
    }
  }
  \newbibmacro*{bbx:parunit}{%
    \ifbibliography
      {\setunit{\bibpagerefpunct}\newblock
       \usebibmacro{pageref}%
       \clearlist{pageref}%
       \setunit{\adddot\par\nobreak}}
      {}
  }
  \renewbibmacro*{doi+eprint+url}{%
    \usebibmacro{bbx:parunit}
    \iftoggle{bbx:doi}
      {\printfield{doi}}
      {}%
    \iftoggle{bbx:eprint}
      {\usebibmacro{eprint}}
      {}%
    \iftoggle{bbx:url}
      {\usebibmacro{url+urldate}}
      {}
  }
  \renewbibmacro*{eprint}{%
    \usebibmacro{bbx:parunit}
    \iffieldundef{eprinttype}
      {\printfield{eprint}}
      {\printfield[eprint:\strfield{eprinttype}]{eprint}}
  }
  \renewbibmacro*{url+urldate}{%
    \usebibmacro{bbx:parunit}
    \printfield{url}%
    \iffieldundef{urlyear}
      {}
      {\setunit*{\addspace}%
       \printtext[urldate]{\printurldate}}
  }
}{}\makeatother

\declaretheorem[numberwithin=section,refname={theorem,theorems},Refname={Theorem,Theorems}]{theorem}
\declaretheorem[sibling=theorem,style=definition]{definition}
\declaretheorem[sibling=theorem,name=Lemma]{lemma}
\declaretheorem[sibling=theorem,name=Proposition]{proposition}

\declaretheorem[sibling=theorem,style=definition,name=Example]{example}

\declaretheorem[numbered=no,name=Question]{question}

\makeatletter\@ifpackageloaded{hyperref}{%
  \usepackage{xcolor}
  \definecolor{dark-red}{rgb}{0.4,0.15,0.15}
  \definecolor{dark-blue}{rgb}{0.15,0.15,0.4}
  \definecolor{medium-blue}{rgb}{0,0,0.5}
  \hypersetup{
    colorlinks,
    linkcolor={dark-red},
    citecolor={dark-blue},
    urlcolor={medium-blue}%
  }

}{}\makeatother

\usepackage{sectsty}
\sectionfont{\large\bfseries}
\subsectionfont{\normalsize\bfseries}

\providecommand{\abs}[1]{\lvert#1\rvert}
\providecommand{\floor}[1]{\lfloor#1\rfloor}
\providecommand{\ceil}[1]{\lceil#1\rceil}
\newcommand{\infw}[1]{%
  \ifcat\noexpand#1\relax\bm{#1}
  \else\mathbf{#1}\fi}          
\newcommand{\N}{\mathbb{N}}

\newcommand{\aci}[1]{\mathcal{A}(#1)}
\newcommand{\acii}[1]{\mathcal{A}_\infty(#1)}
\newcommand{\OL}[1]{\overline{#1}}

\newcommand{\keywords}[1]{\par\noindent{\footnotesize{\em Keywords\/}: #1}}

\begin{document}
  \title{Avoiding abelian powers cyclically}
  \author[,1,2,3]{Jarkko Peltomäki\footnote{Corresponding author.\\E-mail addresses: \href{mailto:r@turambar.org}{r@turambar.org} (J. Peltomäki), \href{mailto:mawhit@mpi-sws.org}{mawhit@mpi-sws.org} (M. A. Whiteland).}}
  \affil[1]{The Turku Collegium for Science and Medicine TCSM, University of Turku, Turku, Finland}
  \affil[2]{Turku Centre for Computer Science TUCS, Turku, Finland}
  \affil[3]{University of Turku, Department of Mathematics and Statistics, Turku, Finland}
  \affil[4]{Max Plank Institute for Software Systems, Saarbrücken, Germany}
  \author[4]{Markus A. Whiteland}
  \date{}
  \maketitle
  \noindent
  \hrulefill
  \begin{abstract}
    \vspace{-1em}
    \noindent
    We study a new notion of cyclic avoidance of abelian powers. A finite word $w$ avoids abelian $N$-powers cyclically
    if for each abelian $N$-power of period $m$ occurring in the infinite word $w^\omega$, we have $m \geq |w|$. Let
    $\mathcal{A}(k)$ be the least integer $N$ such that for all $n$ there exists a word of length $n$ over a $k$-letter
    alphabet that avoids abelian $N$-powers cyclically. Let $\mathcal{A}_\infty(k)$ be the least integer $N$ such that
    there exist arbitrarily long words over a $k$-letter alphabet that avoid abelian $N$-powers cyclically.
    
    We prove that $5 \leq \mathcal{A}(2) \leq 8$, $3 \leq \mathcal{A}(3) \leq 4$, $2 \leq \mathcal{A}(4) \leq 3$, and
    $\mathcal{A}(k) = 2$ for $k \geq 5$. Moreover, we show that $\mathcal{A}_\infty(2) = 4$,
    $\mathcal{A}_\infty(3) = 3$, and $\mathcal{A}_\infty(4) = 2$.
    \vspace{1em}
    \keywords{abelian equivalence, abelian power, abelian power avoidance, cyclic abelian power avoidance, circular word}
    \vspace{-1em}
  \end{abstract}
  \hrulefill
  
  \section{Introduction}
  Ever since the seminal work of A. Thue \cite{1906:uber_unendliche_zeichenreihen}, repetitions or repetition avoidance
  in infinite words has been a central theme in the field of combinatorics on words. Thue showed that there exists a
  ternary word which avoids squares, in symbols $xx$, that is, two identical blocks occurring adjacently in the word.
  Further, he showed that there exists a binary word avoiding cubes, i.e., factors of the form $xxx$. These results are
  best possible concerning integral powers in terms of the size of the underlying alphabet. Thue's results have
  inspired numerous papers on avoiding powers culminating in the papers by Currie and Rampersad
  \cite{2011:a_proof_of_dejeans_conjecture} and Rao \cite{2011:last_cases_of_dejeans_conjecture} proving Dejean's
  conjecture on repetition thresholds.

  An extremely prominent topic in combinatorics on words is the abelian equivalence of words. Two words $u$ and $v$ are
  \emph{abelian equivalent}, in symbols $u \sim v$, if each letter of the underlying alphabet occurs equally many times
  in both words. This concept leads to that of an abelian power: an abelian $N$-power of period $m$ is a word
  $u_0 u_1 \dotsm u_{N-1}$ such that $u_0 \sim u_1 \sim \ldots \sim u_{N-1}$ and the words $u_i$ have common length
  $m$. Thus avoidance of abelian squares, abelian cubes, etc. can be considered. Erd\H{o}s suggested in 1957 the
  problem of whether abelian squares are avoidable on four letters \cite{1957:some_unsolved_problems}. Thus the problem
  of searching for an alphabet of minimal cardinality over which an infinite word avoiding abelian squares exists was
  initiated. Evdokimov \cite{1964:strongly_asymmetric_sequences_generated_finite} gave the first upper bound of $25$.
  Later, Pleasants \cite{1970:non-repetitive_sequences} lowered the bound to five, and finally in 1992 Keränen
  \cite{1992:abelian_squares_are_avoidable_on_4_letters} answered Erd\H{o}s's question in the positive by constructing
  an appropriate infinite word over four letters. Similar questions were considered for smaller alphabets but
  higher-order powers: Dekking \cite{1979:strongly_non-repetitive_sequences_and_progression-free_sets} showed that
  there exist an infinite binary word which avoids abelian fourth powers and an infinite ternary word avoiding abelian
  cubes. Dekking's and Keränen's results are optimal: any binary word of length $10$ contains an abelian cube, and any
  ternary word of length $8$ contains an abelian square.

  There are many variations of the study of avoidance of abelian powers. One direction is to consider avoidance in
  partial words \cite{2012:abelian_repetitions_in_partial_words,2012:avoiding_abelian_squares_in_partial_words}.
  Another is to consider abelian powers occurring in words belonging to specific word classes; see, e.g.,
  \cite{2013:abelian_powers_in_paper-folding_words,2016:abelian_powers_and_repetitions_in_sturmian_words}. Finally, the
  very notion of abelian equivalence can be generalized. See
  \cite{2013:on_a_generalization_of_abelian_equivalence_and_complexity_of_infinite,diss:markus_whiteland,2015:another_generalization_of_abelian_equivalence_binomial}
  for research on $k$-abelian equivalence and binomial equivalence. Of course, there is research on abelian equivalence
  beyond avoidance. We refer the reader to the recent survey \cite{2019:abelian_properties_of_words}. Abelian
  equivalence has also been studied on graphs: the study of abelian $2$-power-free graph colorings has been initiated
  under the term anagram-free colorings in
  \cite{2018:anagram-free_colourings_of_graphs,2018:anagram-free_graph_colouring}. A coloring of a graph is 
  anagram-free if no sequence of colors corresponding to a path in the graph is an abelian $2$-power. We remark that
  anagram-free colorings of cycles correspond to circular avoidance of abelian $2$-powers (see below).
  
  A notion related to this paper is the notion of \emph{circular avoidance}. A word $w$ avoids $N$-powers
  \emph{circularly} if no word in the conjugacy class of $w$ contains an $N$-power as a factor. This is a more
  restrictive type of avoidance and more difficult to study because the language of words avoiding $N$-powers
  circularly is not closed under taking factors. This results in interesting phenomena. For example, Currie shows in
  \cite{2002:there_are_ternary_circular_square-free} that there exists a ternary word of length $n$ avoiding squares
  circularly for $n \geq 18$, but no such word of length $17$ exists. For more on this notion, see, e.g.,
  \cite{2019:circular_repetition_thresholds_small_alphabets} and references therein. According to our knowledge no
  research on the abelian analogue of circular avoidance exists.

  This paper introduces a stronger form of the circular avoidance called cyclic avoidance, and we mainly study it in
  the abelian setting. A word $w$ avoids abelian $N$-powers \emph{cyclically} if any abelian $N$-power occurring in
  $w^\omega = w w \dotsm$ has period at least the length $\abs{w}$ of $w$. The difference between circular and cyclic
  avoidance is that, in cyclic avoidance, periods up to length $\abs{w} - 1$ are disallowed while in circular avoidance
  only periods up to $\floor{\abs{w}/N}$ are disallowed in $w^\omega$. Cyclic avoidance of abelian powers was
  introduced in the recent paper \cite{2020:all_growth_rates_of_abelian_exponents} by the authors of this paper. There
  it served as a tool to construct infinite words with prescribed growth rate of abelian exponents. Due to the
  different focus, the abelian cyclic avoidance was only briefly studied in
  \cite{2020:all_growth_rates_of_abelian_exponents} and only to the extent that was necessary for the main result of
  that paper. The purpose of this paper is to extend this preliminary research by considering the question of what is
  the least number of letters required to avoid abelian $N$-powers cyclically.

  Let $\aci{k}$ be the least integer $N$ such that for all $n$ there exists a word of length $n$ over a $k$-letter
  alphabet that avoids abelian $N$-powers cyclically. Similarly, let $\acii{k}$ be the least integer $N$ such that
  there exist arbitrarily long words over a $k$-letter alphabet that avoid abelian $N$-powers cyclically. The main
  results of this paper are as follows.

  \begin{theorem}\label{thm:main}
    We have $5 \leq \aci{2} \leq 8$, $3 \leq \aci{3} \leq 4$, $2 \leq \aci{4} \leq 3$, and $\aci{k} = 2$ for
    $k \geq 5$.
  \end{theorem}
  
  The lower bound for $\aci{2}$ might be a bit surprising at first sight. However, it can be checked that no binary
  word of length $8$ avoids abelian $4$-powers cyclically. The bounds for $\aci{3}$, $\aci{4}$, and $\aci{5}$ are quite
  straightforward from the results of Dekking and Keränen mentioned previously, but the upper bound for $\aci{2}$
  requires an explicit construction.

  Extending the results of \cite{2020:all_growth_rates_of_abelian_exponents}, we prove the following theorem.

  \begin{theorem}\label{thm:main2}
    We have $\acii{2} = 4$, $\acii{3} = 3$, and $\acii{4} = 2$.
  \end{theorem}

  The last result of the theorem can be seen as progress in resolving a conjecture appearing in
  \cite[p.~17]{2018:anagram-free_graph_colouring}, which reads: for all but finitely many $n$, there exists a
  four-letter word of length $n$ avoiding abelian $2$-powers circularly.

  We also extend previous results of Aberkane and Currie
  \cite{2002:there_are_ternary_circular_square-free,2004:there_exist_binary_circular_5/2+} concerning the circular
  avoidance of powers to our cyclic setting. The results are as follows.

  \begin{theorem}
    If $n \notin \{5, 7, 9, 10, 14, 17\}$ then there exists a word of length $n$ over a $3$-letter alphabet that avoids
    $2$-powers cyclically.
  \end{theorem}

  \begin{theorem}\label{thm:ordinary_binary}
    For each $n$, there exists a word of length $n$ over a $2$-letter alphabet that avoids $5/2^+$-powers cyclically.
  \end{theorem}
  
  The paper is structured as follows. In \autoref{sec:preliminaries}, we introduce notation and the main notions. We
  develop preliminary properties of cyclic abelian repetitions, and recall relevant results from the literature. In
  \autoref{sec:arbitrarily_long} we prove \autoref{thm:main2}. The binary and ternary cases were proved already in our
  previous work. The case of the four letter alphabet requires some technical developments. In
  \autoref{sec:bounds_for_aci}, we prove \autoref{thm:main}. The nonbinary results follow quite straightforwardly from
  results in the literature when combined with our observations in \autoref{sec:arbitrarily_long}. The upper bound for
  the binary case requires an involved construction, splitting into even and odd length words, and is the main
  technical part of the section. In \autoref{sec:ordinary_powers}, we extend known results on circular avoidance of
  ordinary powers to our cyclic setting. We then conclude with future directions of research in
  \autoref{sec:conclusions}.
  
  \section{Preliminaries}\label{sec:preliminaries}
  We use standard terminology and notation of combinatorics on words; see
  \cite{1983:combinatorics_on_words,2002:algebraic_combinatorics_on_words} for standard references. Let $A$ be an
  \emph{alphabet}, that is, a finite set of \emph{letters}, or \emph{symbols}. A \emph{word} over the alphabet $A$ is a
  sequence of letters of $A$ obtained by concatenation. We denote the empty word by $\varepsilon$. The length of a word
  $w$ is denoted by $\abs{w}$, and the symbol $\abs{w}_a$ stands for the number of occurrences of the letter $a$ in
  $w$. If $u$ and $v$ are two words, then we denote their concatenation by $uv$. If $w = uzv$, then $z$ is a
  \emph{factor} of $w$. If $u = \varepsilon$ (resp. $v = \varepsilon$), then $z$ is a \emph{prefix} (resp.
  \emph{suffix}) of $w$. A word $z$ is a \emph{proper prefix} (resp. \emph{proper suffix}) of $w$ if $z$ is a prefix
  (resp. suffix) of $w$ and $z \neq \varepsilon$ and $z \neq w$. If $z$ is a factor of $w$, then we say that $z$
  \emph{occurs} in $w$. If $w = uv$, then by $u^{-1} w$ and $w v^{-1}$ we respectively mean the words $v$ and $u$. If
  $w = u u \dotsm u$ where $u$ is repeated $N$ times, we write $w = u^N$ and say that $w$ is an (ordinary)
  \emph{$N$-power} of period $\abs{u}$. A \emph{fractional power} with exponent $R$, $R > 1$, is a word of the form
  $x^N x'$, where $x'$ is a prefix of $x$ and $R = N + \abs{x'}/\abs{x}$. The set of all words over $A$ is denoted by
  $A^*$. A \emph{language} is a subset of $A^*$. A word $w$ is \emph{primitive} if $w = u^n$ only when $n = 1$. If
  there exist words $x$ and $y$ such that $u = xy$ and $v = yx$, then we say that $u$ and $v$ are \emph{conjugate}. If
  $w = a_0 a_1 \dotsm a_{n-1}$, $a_i \in A$, then the \emph{reversal} of $w$ is the word $a_{n-1} \dotsm a_1 a_0$.

  An infinite word $\infw{w}$ is a mapping from $\N \to A$ (we index words from $0$). We refer to infinite words in
  boldface symbols. We denote the infinite repetition of a finite word $u$ by $u^\omega$.

  Let us define the \emph{Parikh mapping} $\psi\colon A^* \to \N^{\abs{A}}$ by setting
  $\psi(w) = (\abs{w}_a)_{a \in A}$. We refer to the vector $\psi(w)$ as the \emph{Parikh vector} of $w$.

  \begin{definition}
    Let $u, v \in A^*$. We say that $u$ and $v$ are \emph{abelian equivalent} if $\psi(u) = \psi(v)$.
  \end{definition}

  The following definition generalizes $N$-powers.

  \begin{definition}
    Let $u_0$, $\ldots$, $u_{N-1}$, $N \geq 2$, be abelian equivalent and nonempty words of common length $m$. Then
    their concatenation $u_0 \dotsm u_{N-1}$ is an \emph{abelian $N$-power} of period $m$ and exponent $N$. If a word
    (finite or infinite) $w$ does not contain as factors abelian $N$-powers, then we say that $w$ \emph{avoids abelian
    $N$-powers} or that $w$ is \emph{abelian $N$-free}.
  \end{definition}

  The next definition is central to this paper.

  \begin{definition}
    Let $w$ be a word. Then $w$ avoids abelian $N$-powers \emph{cyclically} if for each abelian $N$-power of period $m$
    occurring in the infinite word $w^\omega$, we have $m \geq \abs{w}$.
  \end{definition}

  \begin{example}
    Let $w = 1000100$. Then both $w$ and $w^2$ avoid abelian $5$-powers. However, the word $w^3$ has the abelian
    $5$-power $100 \cdot 010 \cdot 010 \cdot 001 \cdot 001$ of period $3$ as a prefix. Therefore $w$ does not avoid
    abelian $5$-powers cyclically. Since $w^4$ contains an abelian $6$-power of period $4$ beginning from the second
    letter, the word $w$ does not avoid abelian $6$-powers cyclically either. By a straightforward inspection, it can
    be seen that it avoids abelian $7$-powers cyclically. This fact is immediate from the example following the next
    lemma.
  \end{example}

  Notice that there might not be an integer $N$ such that a word $w$ avoids abelian $N$-powers cyclically. This happens
  when, e.g., $w$ is conjugate to an abelian power. The following result characterizes this situation.
  
  \begin{lemma}\label{lem:avoidance_characterization}
    A word $w$ avoids abelian $\abs{w}$-powers cyclically if and only if for each $k < \abs{w}$, $w^k$ is not conjugate
    to an abelian power with period less than $\abs{w}$. Further, if $w$ does not avoid abelian $\abs{w}$-powers
    cyclically, then it does not avoid abelian $N$-powers cyclically for any $N$.
  \end{lemma}
  \begin{proof}
    If $w$ is such that $w^k$ is conjugate to an abelian power of period $m$ with $m < \abs{w}$, then it is immediate that $w$ does not
    avoid abelian $N$-powers cyclically for any $N$. Suppose that $w$ is such that for each $k < \abs{w}$, the word
    $w^k$ is not conjugate to an abelian power with period less that $\abs{w}$. Consider an abelian $\abs{w}$-power
    $u_0 \dotsm u_{\abs{w}-1}$ of period $m$ occurring in $w^\omega$. By conjugating $w$ if necessary, we may assume
    that $u_0 \dotsm u_{\abs{w}-1}$ is a prefix of $w^\omega$. Let $\ell = m\abs{w}/\gcd(m, \abs{w})$ so that
    $\smash[t]{u_0 \dotsm u_{\ell/m-1} = w^{\ell/\abs{w}}}$. The assumption then implies that $m \geq \abs{w}$ or
    $\ell/\abs{w} \geq \abs{w}$. The latter also implies that $m \geq \abs{w}$, so $w$ avoids abelian $\abs{w}$-powers
    cyclically.
  \end{proof}

  The previous example shows that the exponent $\abs{w}$ in the above characterization is tight. We apply the above
  characterization to a subclass of words avoiding abelian $\abs{w}$-powers cyclically.

  \begin{example}
    Let $w$ be a word over $A$ with $\gcd(\{\abs{w}_a : a \in A\}) = 1$. This is satisfied, e.g., when $w$ is not a
    power of a letter and $\abs{w}$ is a prime number. We claim that $w$ avoids abelian $\abs{w}$-powers cyclically. If
    $w^k$ is conjugate to an abelian $N$-power $u_0 \dotsm u_{N-1}$ with $N > k$, then
    \begin{equation*}
       k = \gcd(\{\abs{w^k}_a : a \in A\}) = \gcd(\{N \abs{u_0}_a : a \in A\}) = N \gcd(\{\abs{u_0}_a : a \in A\}) \geq N > k
    \end{equation*}
    which is impossible. Thus if $w^k$ is conjugate to an abelian power, the period of this abelian power must be at
    least $\abs{w}$. The claim follows from \autoref{lem:avoidance_characterization}.

    The condition $\gcd(\{\abs{w}_a : a \in A\}) = 1$ is not necessary: the word $001122$ avoids abelian $3$-powers
    cyclically and so avoids abelian $6$-powers cyclically.
  \end{example}

  The following lemma is elementary, but it simplifies the arguments in the rest of the paper drastically.

  \begin{lemma}\label{lem:reduction_short}
    Assume that $x^\omega$ contains an abelian $N$-power of period $m$ with $\frac12 \abs{x} \leq m < \abs{x}$. Then
    it contains an abelian $N$-power with period $\abs{x} - m$.
  \end{lemma}
  \begin{proof}
    There is nothing to prove when $m = \tfrac12 \abs{x}$, so we may assume that $m > \tfrac12 \abs{x}$. Without loss
    of generality, we may further assume that $x^\omega$ begins with an abelian $N$-power $u_0 \dotsm u_{N-1}$. We
    show, by induction on $N$, that if $x^{\omega}$ begins with an abelian $N$-power with period $m$, then
    the word $x^{N-1}$ ends with an abelian $N$-power $s_{N-1} \dotsm s_0$ of period
    $\abs{x} - m$.

    Consider first the base case $N = 2$. Since $m$ satisfies $\tfrac12 \abs{x} < m < \abs{x}$, we have
    $\abs{u_0} < \abs{x} < \abs{u_0 u_1}$. We may write $x = u_0 s_0$ and $u_1 = s_0 p$, where $s_0$ is the length
    $\abs{x} - m$ suffix of $x$ and $p$ is a prefix of $x$. Notice that $\abs{p} < m$, so we have $u_0 = p s_1$ for the
    suffix $s_1$ of $u_0$ of length $\abs{s_0}$. We find that
    \begin{equation*}
      0 = \psi(u_0) - \psi(u_1) = \psi(p s_1) - \psi(s_0 p) = \psi(s_1) - \psi(s_0).
    \end{equation*}
    Thus $s_1$ is abelian equivalent to $s_0$, and $x$ ends with the abelian $2$-power $s_1 s_0$.

    Let then $N > 2$. By proceeding as in the base case, we find that $x$ ends with the abelian $2$-power $s_1 s_0$ of
    period $\abs{x}-m$. Consider the conjugate $z = s_0 u_0$ of $x$: the word $z^\omega$ begins with the abelian power
    $u_1 \dotsm u_{N-1}$. By the induction hypothesis, $z^{N-2}$ ends with the abelian power $s_{N-1} \dotsm s_1$ of
    period $\abs{x}-m$. To conclude the proof, we notice that $x^{N-1} = u_0 z^{N-2} s_0$. The claim follows.
  \end{proof}

  \section{Values of \texorpdfstring{$\acii{k}$}{A\_infty(k)}}\label{sec:arbitrarily_long}
  The aim of this section is to prove \autoref{thm:main2}. Recall that $\acii{k}$ is the least $N$ such that there
  exist arbitrarily long words over a $k$-letter alphabet that avoid abelian $N$-powers cyclically. Our constructions
  for the main result of the section involve building arbitrarily long words with morphisms. Next we recall the
  definition of a morphism and related abelian avoidance results.
  
  A \emph{morphism} $\sigma\colon A^* \to B^*$ is a mapping such that $\sigma(uv) = \sigma(u)\sigma(v)$ for all words
  $u, v \in A^*$. The morphism $\sigma$ is \emph{prolongable} on a letter $a$ if $\sigma(a)$ has prefix $a$ and
  $\lim_{n\to\infty} \abs{\sigma^n(a)} = \infty$. Thus iterating $\sigma$ on the letter $a$ produces an infinite word
  that is a fixed point of $\sigma$. We denote this fixed point by $\sigma^\omega(a)$. The set
  $\{w\colon \text{$w$ is a factor of $\sigma^n(a)$ for some $n \geq 0$ and $a \in A$}\}$ is called the language of the
  morphism $\sigma$.

  \begin{definition}
    A morphism $\sigma\colon A^* \to B^*$ is \emph{abelian $N$-free} if $\sigma(w)$ is abelian $N$-free for all abelian
    $N$-free words $w$ in $A^*$.
  \end{definition}

  Notice that $\abs{\sigma(a)} \geq 1$ for each letter $a$ when $\sigma$ is a prolongable abelian $N$-free morphism.
  Indeed, if $\sigma(a) = \varepsilon$ and $\sigma(b) \neq \varepsilon$, then the abelian $N$-free word $bab^{N-1}$ has
  an abelian $N$-power in its image.

  There are several results in the literature concerning abelian $N$-free morphisms. For example, Dekking gave
  sufficient conditions for a morphism to be abelian $N$-free in
  \cite{1979:strongly_non-repetitive_sequences_and_progression-free_sets}. Later, Carpi extended the results of Dekking
  by giving sharper sufficient conditions for a morphism to be abelian $N$-free in
  \cite{1993:on_abelian_power-free_morphisms}. It is worth mentioning that Carpi's general conditions, for abelian
  $N$-freeness, are necessary and sufficient once the domain alphabet has cardinality at least $6$
  \cite[Proposition~2]{1993:on_abelian_power-free_morphisms}. It remains open to this day, whether the conditions
  characterize abelian $N$-free morphisms for smaller domain alphabets.

  Let us recall some morphisms that are abelian $N$-free for small values of $N$. The first is the morphism $\sigma_3$,
  found in \cite{2004:the_number_of_binary_words_avoiding_abelian}, that satisfies Dekking's conditions for the
  exponent $4$. The morphism $\sigma_3$ is different from the morphism of
  \cite[Thm.~1]{1979:strongly_non-repetitive_sequences_and_progression-free_sets}. We use this different morphism to
  reduce the amount of computations required to prove \autoref{thm:binary_case}. See
  \cite[Example~2]{2004:the_number_of_binary_words_avoiding_abelian} for the proof of the following lemma.

  \begin{lemma}\label{lem:modified_dekking}
    The morphism $\sigma_3\colon 0 \mapsto 0001$, $1 \mapsto 101$ satisfies Dekking's conditions for the exponent $4$.
    It is thus abelian $4$-free.
  \end{lemma}
  %
  
  By \cite[Thm.~2]{1979:strongly_non-repetitive_sequences_and_progression-free_sets}, the morphism
  $\sigma_4\colon 0 \mapsto 0012, 1 \mapsto 112, 2 \mapsto 022$ satisfies Dekking's conditions for the exponent $3$ and
  is thus abelian $3$-free. The above two morphisms are prolongable on the letter $0$. It thus follows that the
  infinite words $\sigma_3^{\omega}(0)$ and $\sigma_4^{\omega}(0)$ avoid abelian $4$-powers and $3$-powers
  respectively.

  Let $\pi\colon 0 \mapsto 1, 1 \mapsto 2, 2 \mapsto 3, 3 \mapsto 0$. Consider the morphism
  $\phi\colon \{0, 1, 2, 3\}^* \to \{0, 1, 2, 3\}^*$ defined by setting
  \begin{align*}
    \phi(0) &= 0120232123203231301020103101213121021232021 \cdot \\
            &\hspace{1.45em} 013010203212320231210212320232132303132120, \\
    \phi(1) &= \pi(\phi(0)), \\
    \phi(2) &= \pi(\phi(1)), \\
    \phi(3) &= \pi(\phi(2)).
  \end{align*}
  Keränen proved in the breakthrough paper \cite{1992:abelian_squares_are_avoidable_on_4_letters} that the fixed point
  $\phi^\omega(0)$ of the morphism $\phi$ is abelian $2$-free. See also his more recent paper
  \cite{2009:a_powerful_abelian_square-free_substitution} for additional morphisms with this property. Carpi simplified
  Keränen's proof in \cite{1993:on_abelian_power-free_morphisms} by showing that it satisfies Carpi's conditions for
  the exponent $2$, and it is thus abelian $2$-free. It is noteworthy that the morphism does not satisfy Dekking's
  conditions for the exponent $2$. 
  

  Let us prove a general result related to abelian $N$-power cyclical avoidance and abelian $N$-free morphisms.

  \begin{proposition}\label{prp:preserves_N_circularity}
    Let $\sigma\colon A^* \to B^*$ be an abelian $N$-free morphism, and assume that $w$ in $A^*$ is a word that avoids
    abelian $N$-powers cyclically. If $N > 2$, then $\sigma(w)$ avoids abelian $N$-powers cyclically. If $N=2$ and
    $\abs{w}\geq 2$, then $\sigma(w)$ avoids abelian $2$-powers cyclically.
    \end{proposition}
  \begin{proof}
    Suppose for a contradiction that $\sigma(w)$ does not avoid abelian $N$-powers cyclically. Assume thus that
    $u_0 \dotsm u_{N-1}$ is an abelian $N$-power occurring in $\sigma(w)^\omega$ with $\abs{u_0} < \abs{\sigma(w)}$. By
    \autoref{lem:reduction_short}, we may assume that $\abs{u_0} \leq \floor{\abs{\sigma(w)}/2}$, so
    \begin{equation*}
      \abs{u_0 \dotsm u_{N-1}} \leq N \floor{\abs{\sigma(w)}/2} \leq N \abs{\sigma(w)}/2 \leq \ceil{N/2} \abs{\sigma(w)}.
    \end{equation*}
    We conclude that $u_0 \dotsm u_{N-1}$ is a factor of $\sigma(w)^{\ceil{N/2} + 1}$. Let $a_0 \dotsm a_{\ell-1}$,
    $a_i \in A$, be a factor of $w^{\lceil N/2 \rceil + 1}$ of minimal length for which $\sigma(a_0 \dotsm a_{\ell-1})$
    contains $u_0 \dotsm u_{N-1}$. We may write $\sigma(a_0 \dotsm a_{\ell-1}) = p_0 u_0 \dotsm u_{N-1} s_{\ell-1}$
    with $\sigma(a_0) = p_0 s_0$ and $\sigma(a_{\ell-1}) = p_{\ell-1} s_{\ell-1}$. Since $\sigma$ is abelian $N$-free,
    it follows that $a_0 \dotsm a_{\ell-1}$ contains an abelian $N$-power $v_0 \dotsm v_{N-1}$. As $w$ avoids abelian
    $N$-powers cyclically, we have $\abs{v_0} \geq \abs{w}$. Therefore the word $v_i$ has a conjugate of $w$ as a
    factor, so $\abs{\sigma(v_i)} \geq \abs{\sigma(w)}$ for all $i$. Thus
    \begin{equation*}
      N\abs{\sigma(w)} \leq \abs{\sigma(v_0 \dotsm v_{N-1})} \leq \abs{\sigma(a_0 \dotsm a_{\ell-1})} \leq \abs{\sigma(w)^{\ceil{N/2} + 1}} = (\ceil{N/2} + 1)\abs{\sigma(w)}.
    \end{equation*}
    This inequality holds only when $N \leq 3$ in which case equality is forced. For $N\geq 4$, this contradiction
    suffices for the claim. For the remainder of the proof, we operate under the assumption $N \leq 3$. Observe that
    the above computation shows that
    $\abs{\sigma(v_0 \dotsm v_{N-1})} = \abs{\sigma(a_0 \dotsm a_{\ell-1})} = N\abs{\sigma(w)}$. It follows that
    $v_i = w$ for all $i$ and $w^N = a_0 \dotsm a_{\ell-1}$.

    We claim that either $\abs{p_0 u_0} \geq \abs{\sigma(w)}$ or $\abs{u_{N-1} s_{\ell-1}} \geq \abs{\sigma(w)}$.
    Indeed, this is clear if $N = 2$ and if $N = 3$ and $\abs{p_0 u_0}, \abs{u_2 s_{\ell-1}} < \abs{\sigma(w)}$, then
    $\abs{u_1} > \abs{\sigma(w)}$ contrary to our assumptions. We assume that $\abs{p_0 u_0} \geq \abs{\sigma(w)}$; the
    other case is symmetric.

    Next we claim that $\abs{p_0 u_0 u_1} \leq 2\abs{\sigma(w)}$. If not, then
    $\abs{p_0} > \abs{\sigma(w)} \geq \abs{\sigma(a_0)}$ because
    $\abs{u_0 u_1} \leq 2 \floor{\abs{\sigma(w)/2}} \leq \abs{\sigma(w)}$ by our assumption. Since $p_0$ is a prefix of $\sigma(a_0)$,
    this is impossible. We may thus write $\sigma(w) = p u_1 s$ in such a way that $p_0 u_0 = \sigma(w) p$.

    Observe that $\psi(u_0) = \psi(\sigma(w)) - \psi(p_0) + \psi(p)$ and
    $\psi(u_1) = \psi(\sigma(w)) - \psi(p) - \psi(s)$. Since $\psi(u_0) = \psi(u_1)$, we conclude that
    $\psi(p_0) - \psi(p) = \psi(p) + \psi(s)$. Since the Parikh vector $\psi(p) + \psi(s)$ has nonnegative entries, we 
    see that $p$ is a prefix of $p_0$ (both words are prefixes of $\sigma(w)$). We conclude that the words
    $sp$ and $p^{-1}p_0$ are abelian equivalent. Thus by writing $sp_0 = sp \cdot p^{-1}p_0$, we see that $sp_0$ is an
    abelian $2$-power. Suppose now that $N = 2$. This implies that $s = s_{\ell-1}$, so $s_{\ell-1} p_0$ is an abelian
    $2$-power. Since $s_{\ell-1} p_0$ is a factor of $\sigma(a_{\ell-1} a_0)$, it must be that $a_{\ell-1} = a_0$ as
    $\sigma$ is abelian $2$-free. Therefore $w^\omega$ contains the abelian $2$-power $a_{\ell-1} a_0$ of period $1$.
    Since $w$ avoids abelian $2$-powers cyclically, we infer that $\abs{w} = 1$. This gives the latter claim.

    Suppose finally that $N = 3$. Then $s \sigma(w) = u_2 s_{\ell-1}$. We have
    $\psi(u_2) = \psi(\sigma(w)) + \psi(s) - \psi(s_{\ell-1})$. Since $\psi(u_0) = \psi(u_1) = \psi(u_2)$, we get
    $\psi(p_0) - \psi(p) = \psi(p) + \psi(s) = \psi(s_{\ell-1}) - \psi(s)$. Since $\psi(p) + \psi(s)$ has nonnegative
    entries, we conclude that $s$ is a suffix of $s_{\ell-1}$. We may now write
    $s_{\ell-1} p_0 = s_{\ell-1} s^{-1} \cdot s p \cdot p^{-1}p_0$ and conclude that $s_{\ell-1} p_0$ is an abelian
    $3$-power. Now $s_{\ell-1} p_0$ is a factor of $\sigma(a_{\ell-1} a_0)$, so the image of the abelian $3$-free word
    $a_{\ell-1} a_0$ contains an abelian $3$-power. This contradicts the fact that $\sigma$ is abelian $3$-free. This
    proves the former claim.
  \end{proof}
  
  Notice that for the case $N=2$ in the above proposition, the assumption $\abs{w}\geq 2$ cannot be omitted. Indeed,
  the word $0$ avoids abelian $2$-powers cyclically, but the word $\phi(0)$ does not (here $\phi$ is Keränen's
  morphism). This is evident from the fact that $\phi(0)$ begins and ends with the letter $0$, so $\phi(0)^2$ contains
  the abelian $2$-power $00$.
  
  The fact that $\acii{2} = 4$ and $\acii{3} = 3$ was already established in
  \cite[Thm.~8]{2020:all_growth_rates_of_abelian_exponents}. Our following proof simplifies and unifies the arguments
  due to the above proposition. The main task here is to prove that $\acii{4} = 2$, and we do this by iterating
  Keränen's morphism $\phi$ on suitable words.

  \begin{proof}[Proof of \autoref{thm:main2}]
    Recall the morphisms $\sigma_3$ and $\sigma_4$ defined above. They are abelian $4$-free and abelian $3$-free,
    respectively. Now the word $0$ avoids abelian $N$-powers cyclically for all $N \geq 2$. Thus
    \autoref{prp:preserves_N_circularity} implies that the words in the sequences $(\sigma_3^n(0))_n$ and
    $(\sigma_4^n(0))_n$ avoid abelian $N$-powers cyclically for $N = 4$ and $N = 3$, respectively. As the morphisms are
    prolongable on $0$, this establishes that $\acii{2} = 4$ and $\acii{3} = 3$.

    The word $01$ avoids abelian $2$-powers cyclically. Thus \autoref{prp:preserves_N_circularity} implies that the
    words in the sequence $(\phi^n(01))_n$ avoid abelian $2$-powers cyclically. Therefore $\acii{4} = 2$ as $\phi$ is
    prolongable on $0$.
  \end{proof}

  \section{Bounds for \texorpdfstring{$\aci{k}$}{A(k)}}\label{sec:bounds_for_aci}
  Recall that $\aci{k}$ is the least $N$ such that for all $n$ there exists a word of length $n$ over a $k$-letter
  alphabet that avoids abelian $N$-powers cyclically. This section is devoted to proving \autoref{thm:main}. When
  $k \geq 3$, the idea is simply to add a new letter to a word avoiding abelian $N$-powers cyclically. For $k = 2$,
  this idea does not work, and we provide an explicit construction of the required words.

  \begin{lemma}\label{lem:add_letter}
    Let $w$ be a word that avoids abelian $N$-powers and $\#$ a letter that does not appear in $w$. Then the word
    $w \#$ avoids abelian $N$-powers cyclically.
  \end{lemma}
  \begin{proof}
    Set $\infw{w} = (w \#)^\omega$, and assume for a contradiction that an abelian $N$-power $u_0 \dotsm u_{N-1}$ such
    that $\abs{u_0} < \abs{w \#}$ occurs in $\infw{w}$. By \autoref{lem:reduction_short}, we may assume that
    $\abs{u_0} \leq \tfrac12 \abs{w \#}$. Thus $\abs{u_0 u_1} \leq \abs{w \#}$ and $\#$ can occur in $u_0 u_1$ at most
    once. Thus $\#$ does not occur in $u_0$, and so $u_0 \dotsm u_{N-1}$ must be a factor of $w$. This contradicts the
    assumption that $w$ avoids abelian $N$-powers.
  \end{proof}

  \begin{theorem}\label{thm:main_3_4_5}
    We have $3 \leq \aci{3} \leq 4$, $2 \leq \aci{4} \leq 3$, and $\aci{k} = 2$ for $k \geq 5$.
  \end{theorem}
  \begin{proof}
    It is straightforward to verify that every ternary word of length $8$ contains an abelian $2$-power, so
    $\aci{3} \geq 3$. Obviously $\aci{k} \geq 2$ for $k \geq 4$.

    Recall the abelian $4$-free morphism $\sigma_3$ from \autoref{lem:modified_dekking}. Taking $w$ to be a factor of
    $\sigma_3^\omega(0)$ of length $n - 1$, we see by an application of \autoref{lem:add_letter} that the word $w \#$
    of length $n$ over the alphabet $\{0, 1, \#\}$ avoids abelian $4$-powers cyclically. In addition, the word $0$
    avoids abelian $4$-powers cyclically, so $\aci{3} \leq 4$.

    The morphisms $\sigma_4$ and $\phi$, as defined in \autoref{sec:arbitrarily_long}, are abelian $3$-free and abelian
    $2$-free, respectively. Similar to the previous paragraph, we see that $\aci{4} \leq 3$ and $\aci{5} \leq 2$.
  \end{proof}

  Our next aim is to prove the following theorem.

  \begin{theorem}\label{thm:binary_case}
    We have $5 \leq \aci{2} \leq 8$.
  \end{theorem}

  We prove \autoref{thm:binary_case} by explicitly constructing the required words for each length. Our construction is
  inspired by the proof of \cite[Thm.~4]{2012:abelian_repetitions_in_partial_words}. Consider the morphism
  $\sigma\colon 0 \mapsto 0001, 1 \mapsto 101$ of \autoref{lem:modified_dekking} and the prefix $w$ of its fixed point
  $\sigma^\omega(0)$ of length $n$. Let $h\colon 0 \mapsto 1, 1 \mapsto 0$ and $\OL{w}$ be the reversal of $h(w)$. Set
  \begin{align*}
    f   &= \OL{w} \! \diamond \! w, \\
    g_1 &= \OL{w} w, \text{ and} \\
    g_2 &= \OL{w}^\bullet w,
  \end{align*}
  where $\diamond \in \{0, 1\}$ and $\OL{w}^\bullet$ is obtained from $\OL{w}$ by changing its final letter to $0$. We
  further define $\infw{F} = f^\omega$, $\infw{G}_1 = g_1^\omega$, and $\infw{G}_2 = g_2^\omega$. Recall that the words
  $w$ and $\OL{w}$ do not contain abelian $4$-powers as factors. This follows from \autoref{lem:modified_dekking} and
  the discussion following it. Furthermore, $\OL{w}^{\bullet}$ avoids abelian $5$-powers.
  
  In \autoref{ssec:odd_case}, we prove that $f$ avoids abelian $8$-powers cyclically for all $n$.
  \autoref{ssec:even_case} establishes that $g_1$ avoids abelian $8$-powers cyclically if $n$ is odd and $g_2$ avoids
  abelian $8$-powers cyclically when $n$ is even. These results establish that $\aci{2} \leq 8$.
  \autoref{thm:binary_case} follows from the observation that there does not exist a binary word of length $8$ avoiding
  abelian $4$-powers cyclically. However, such a word exists in the circular sense (see the introduction): $00010011$.

  The approach taken in \autoref{ssec:odd_case} is identical to that of \autoref{ssec:even_case}. Several of the
  structural lemmas carry over with very minor modifications. In particular, we encourage the reader to notice that the
  presence of the symbol $\diamond$ does not often play any role. We shall make use of the following notion.
 
  \begin{definition}
    Let $u$ be a binary word over the alphabet $\{0, 1\}$, and define $\Delta(u) = \abs{u}_0 - \abs{u}_1$. If
    $\Delta(u) > 0$ (resp. $\Delta(u) < 0$, $\Delta(u) = 0$), then $u$ is \emph{light} (resp. \emph{heavy},
    \emph{neutral}).
  \end{definition}

  Let us first establish some properties of the fixed point $\sigma^\omega(0)$ of $\sigma$. In particular, we consider
  properties of short factors of $\sigma^\omega(0)$, which can be verified with the help of a computer.

  The word $w$ below refers to the construction of the words $f$, $g_1$, and $g_2$.

  \begin{lemma}\label{lem:light}
    If $u$ is a factor of $w$ such that $\abs{u} \geq 29$, then $u$ is light.
  \end{lemma}
  \begin{proof}
    It is straightforward to check that if $u$ is a factor of the language of $\sigma$ such that
    $29 \leq \abs{u} \leq 2 \times 29 = 58$, then $u$ is light. Any factor of length at least $58$ can be written as a
    concatenation of words of length between $29$ and $58$, so it follows that all factors $u$ with $\abs{u} \geq 29$
    are light.
  \end{proof}

  \begin{lemma}\label{lem:difference_upper_bound}
    If $u$ is a factor of $w$ such that $\abs{u} < 29$, then $\Delta(u) \geq -3$.
  \end{lemma}
  \begin{proof}
    This is a finite check.
  \end{proof}

  \begin{lemma}\label{lem:delta_lower_bound_6}
    If $u$ is a factor of $w$ such that $\abs{u} \geq 64$, then $\Delta(u) \geq 6$.
  \end{lemma}
  \begin{proof}
    Let $u$ be a factor of $w$ such that $\abs{u} \geq 6 \times 29 = 174$ and factorize $u = u_0 \dotsm u_5$ in such a
    way that $\abs{u_i} \geq 29$ for all $i$. Since $\abs{u_i} \geq 29$, we have $\Delta(u_i) > 0$ by
    \autoref{lem:light}. Consequently, we see that $\Delta(u) = \sum_{i=0}^5 \Delta(u_i) \geq 6$. It can be verified
    with the help of a computer that if $64 \leq \abs{u} < 174$, then $\Delta(u) \geq 6$.
  \end{proof}

  \subsection{Odd Length Case}\label{ssec:odd_case}
  The aim of this subsection is to prove the following proposition.

  \begin{proposition}\label{prp:odd_8_power}
    The word $f$ avoids abelian $8$-powers cyclically.
  \end{proposition}

  While the letter $\diamond$ can be freely chosen to be either $0$ or $1$, we use the symbol as a marker in the proofs
  that follow. \autoref{prp:odd_8_power} can be verified to be true when $\abs{w} < 5 \times 29 = 145$. Thus in what
  follows, we implicitly assume that $\abs{w} \geq 145$.

  Let $u_0 \dotsm u_{N-1}$ be an abelian $N$-power occurring in $\infw{F}$, and consider a word $u_i$ for some $i$. We
  classify the word $u_i$ as follows.

  \begin{enumerate}[(A)]
    \item $u_i = \alpha_i \! \diamond \! \beta_i$ for a suffix $\alpha_i$ of $\OL{w}$ and a prefix $\beta_i$ of $w$;
    \item $u_i = \alpha_i \beta_i$ for a nonempty suffix $\alpha_i$ of $w$ and a nonempty prefix $\beta_i$ of $\OL{w}$.
  \end{enumerate}

  In the proofs, we implicitly use the above factorizations using the words $\alpha_i$ and $\beta_i$. Notice that it is
  not necessary for $u_i$ to have type A or B. See \autoref{fig:types} for clarification.

  \begin{figure}
  \centering
  \begin{tikzpicture}
    \tikzstyle{ln}=[line width=0.8pt,black]
    \tikzstyle{arc}=[line width=0.8pt,black]
    \tikzstyle{point}=[line width=0.5pt,black,fill=white]
    \tikzstyle{point2}=[line width=0.8pt,black]

    \newcommand{\sqdiamond}[1][fill=black]{\mathbin{\tikz [x=1ex,y=1ex,line width=.1ex,line join=round] \draw [#1] (0,.5) -- (.5,1) -- (1,.5) -- (.5,0) -- (0,.5) -- cycle;}}
    \tikzmath{\width = 3; \aoffset = 0.5; \awidth = 1.8; \ayoffset = -0.2;}

    \draw[ln,|-|] (0,0) -- (\width,0);
    \draw[ln,-|] (\width,0) -- ({2*\width},0);
    \draw[ln,-|] ({2*\width},0) -- ({3*\width},0);

    \node at ({3*\width + 0.5},0) {\footnotesize $\dotsm$};
    \node at (-0.5,0) {\footnotesize $\mathbf{F}\colon$};
    \node at ({0.5*\width},0) {$\sqdiamond[fill=white]$};
    \node at ({1.5*\width},0) {$\sqdiamond[fill=white]$};
    \node at ({2.5*\width},0) {$\sqdiamond[fill=white]$};

    \node at ({\aoffset + 0.5*\awidth},0.7) {\footnotesize $u_0$};
    \node at ({\aoffset + 1.5*\awidth},0.7) {\footnotesize $u_1$};
    \node at ({\aoffset + 2.5*\awidth},0.7) {\footnotesize $u_2$};

    \draw[arc,-] (\aoffset,0) to[bend left] ({\aoffset + \awidth},0);
    \draw[arc,-] ({\aoffset+\awidth},0) to[bend left] ({\aoffset + 2*\awidth},0);
    \draw[arc,-] ({\aoffset+2*\awidth},0) to[bend left] ({\aoffset + 3*\awidth},0);

    \draw [thick,decoration={brace,mirror},decorate] (\aoffset,\ayoffset) -- ({0.5*\width},\ayoffset) node [midway,align=center,yshift=-10] {\footnotesize $\alpha_0$};
    \draw [thick,decoration={brace,mirror},decorate] ({0.5*\width},\ayoffset) -- ({\aoffset + 1*\awidth},\ayoffset) node [midway,align=center,yshift=-10] {\footnotesize $\beta_0$};
    \draw [thick,decoration={brace,mirror},decorate] ({\aoffset + \awidth},\ayoffset) -- ({1*\width},\ayoffset) node [midway,align=center,yshift=-10] {\footnotesize $\alpha_1$};
    \draw [thick,decoration={brace,mirror},decorate] ({1*\width},\ayoffset) -- ({\aoffset + 2*\awidth},\ayoffset) node [midway,align=center,yshift=-10] {\footnotesize $\beta_1$};
    \draw [thick,decoration={brace,mirror},decorate] ({\aoffset + 2*\awidth},\ayoffset) -- ({1.5*\width},\ayoffset) node [midway,align=center,yshift=-10] {\footnotesize $\alpha_2$};
    \draw [thick,decoration={brace,mirror},decorate] ({1.5*\width},\ayoffset) -- ({\aoffset + 3*\awidth},\ayoffset) node [midway,align=center,yshift=-10] {\footnotesize $\beta_2$};

    \draw [thick,decoration={brace,mirror},decorate] (0,-1) -- ({0.5*\width},-1) node [midway,align=center,yshift=-10] {\footnotesize $\OL{w}$};
    \draw [thick,decoration={brace,mirror},decorate] (0.5*\width,-1) -- ({1*\width},-1) node [midway,align=center,yshift=-10] {\footnotesize $w$};
  \end{tikzpicture}
  \caption{A depiction of the structure of $\infw{F}$. The words $u_0$ and $u_2$ are of type A and $u_1$ is of type B.}
  \label{fig:types}
  \end{figure}
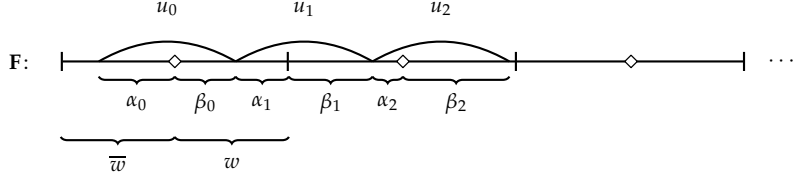

  The following simple observation is very important in the subsequent proofs.

  \begin{lemma}\label{lem:imbalance}
    Suppose that $u$ and $v$ are words of common length such that $\abs{u} \geq 29$. If $u$ is a factor of $w$ and $v$
    is a factor of $\OL{w}$, then $u$ and $v$ are not abelian equivalent.
  \end{lemma}
  \begin{proof}
    If $u$ is a factor of $w$ and $\abs{u} \geq 29$, then $u$ is light by \autoref{lem:light}. If $v$ is a factor of
    $\OL{w}$, then $\OL{v}$ is a factor of $w$ and must thus also be light. This means that $v$ is heavy, so $u$ and
    $v$ cannot be abelian equivalent.
  \end{proof}

  Next we show that any abelian $8$-power occurring in $\infw{F}$ must have a relatively large period.

  \begin{lemma}\label{lem:odd_half}
    If an abelian $8$-power of period $m$ occurs in $\infw{F}$, then $m > \tfrac12 \abs{w}$.
  \end{lemma}
  \begin{proof}
    Assume for a contradiction that $\infw{F}$ contains an abelian $8$-power $u_0 \dotsm u_7$ such that
    $\abs{u_0} \leq \tfrac12 \abs{w}$. There exists $u_i$ such that $u_i$ is of type A or B because $w$ and $\OL{w}$
    avoid abelian $4$-powers. We suppose that $u_i$ is of type A; the case that it is of type B is analogous. Suppose
    first that $\abs{u_0} < 29$. If $i \leq 3$, then $u_{i+1} u_{i+2} u_{i+3} u_{i+4}$ is a factor of $w$ because
    $\abs{w} \geq 5 \times 29 = 145$. This is impossible as $w$ avoids abelian $4$-powers. Thus $i \geq 4$, but then
    $u_{i-4} u_{i-3} u_{i-2} u_{i-1}$ is an abelian $4$-power occurring in $\OL{w}$. We conclude that
    $\abs{u_0} \geq 29$.

    Assume that $1 \leq i \leq 6$, so that $u_{i-1}$ and $u_{i+1}$ exist. Since $\abs{u_0} \leq \tfrac12 \abs{w}$, the
    word $u_{i+1}$ ends before the end of $w$ and the word $u_{i-1}$ begins after the beginning of $\OL{w}$. Therefore
    $u_{i-1}$ is a factor of $\OL{w}$ and $u_{i+1}$ is a factor of $w$. \autoref{lem:imbalance} shows that $u_{i-1}$
    and $u_{i+1}$ cannot be abelian equivalent; a contradiction. Suppose then that $i = 0$. Then $u_1$ is a factor of
    $w$ since $\abs{u_0} \leq \tfrac12 \abs{w}$. Since $w$ avoids abelian $4$-powers, the word $u_1 u_2 u_3 u_4$ cannot
    be a factor of $w$. Thus $u_2$, $u_3$, or $u_4$ is of type B. Consequently, one of the words $u_3$, $u_4$, and
    $u_5$ must be a factor $\OL{w}$. This again contradicts \autoref{lem:imbalance}. The case $i = 7$ is similar.
  \end{proof}


  The following two lemmas are technical lemmas that indicate what values $\Delta(u_i)$ may take for a $u_i$ of type A
  or B depending on the lengths of the corresponding words $\alpha_i$ and $\beta_i$.


  \begin{lemma}\label{lem:split_bounds_1}
    Suppose that the word $\infw{F}$ contains an abelian $N$-power $u_0 \dotsm u_{N-1}$. Say $u_i$ is of type B and
    write $u_i = \alpha_i \beta_i$.
    \begin{enumerate}[(i)]
      \item If $\abs{\alpha_i} \geq \abs{\beta_i}$, then $\Delta(u_i) \geq -3$.
      \item If $\abs{\alpha_i} \leq \abs{\beta_i}$, then $\Delta(u_i) \leq 3$.
    \end{enumerate}
  \end{lemma}
  \begin{proof}
    Suppose that $\abs{\alpha_i} \geq \abs{\beta_i}$. Since $\beta_i$ is a prefix of $\OL{w}$, the word $\OL{\beta_i}$
    is a suffix of $w$. We may thus write $\alpha_i = z \OL{\beta_i}$ for some word $z$. Since
    $\abs{\OL{\beta_i} \beta_i}_0 = \abs{\OL{\beta_i} \beta_i}_1$, we have $\Delta(u_i) = \Delta(z)$. The word
    $z$ is a factor of $w$, so if $\Delta(z) \leq 0$, then $\abs{z} < 29$ by \autoref{lem:light}, and hence
    $\Delta(z) \geq -3$ by \autoref{lem:difference_upper_bound}. Claim (i) follows. Claim (ii) is proved
    symmetrically.
  \end{proof}

  \begin{lemma}\label{lem:split_bounds_2}
    Suppose that the word $\infw{F}$ contains an abelian $N$-power $u_0 \dotsm u_{N-1}$. Say $u_i$ is of type A and
    write $u_i = \alpha_i \! \diamond \! \beta_i$.
    \begin{enumerate}[(i)]
      \item If $\abs{\alpha_i} \geq \abs{\beta_i}$, then $\Delta(u_i) - \Delta(\diamond) \leq 3$.
      \item If $\abs{\alpha_i} \leq \abs{\beta_i}$, then $\Delta(u_i) - \Delta(\diamond) \geq -3$.
    \end{enumerate}
  \end{lemma}
  \begin{proof}
    This proof is similar to that of \autoref{lem:split_bounds_1}. Say $\abs{\alpha_i} \geq \abs{\beta_i}$. Then
    $\OL{\beta_i}$ is a suffix of $\OL{w}$, and we may write $\alpha_i = z \OL{\beta_i}$. Thus
    $\Delta(\alpha_i \beta_i) = \Delta(z)$. If $\Delta(z) \geq 0$, then $\Delta(z) \leq 3$ by
    Lemmas \ref{lem:light} and \ref{lem:difference_upper_bound}. It follows that
    $\Delta(u_i) = \Delta(\alpha_i \! \diamond \! \beta_i) = \Delta(z) + \Delta(\diamond) \leq 3 + \Delta(\diamond)$.
    Claim (ii) is analogous.
  \end{proof}

  We aim to combine \autoref{lem:odd_half} and the following observation. Together they imply that if an abelian
  $8$-power $u_0 \dotsm u_7$ occurs in $\infw{F}$, then each of the factors $u_i$ has type A or type B.

  \begin{lemma}\label{lem:cube_1}
    Let $u_0 u_1 u_2$ be an abelian $3$-power occurring in $\infw{F}$. If
    \begin{enumerate}[(i)]
      \item $u_0$ occurs in $w$ or $\OL{w}$ or
      \item $u_2$ occurs in $w$ or $\OL{w}$,
    \end{enumerate}
    then $\abs{u_0} \leq \tfrac12 \abs{w}$.
  \end{lemma}
  \begin{proof}
    Assume on the contrary that $\abs{u_0} > \tfrac12 \abs{w}$ and $u_0$ occurs in $w$. Now $\abs{u_0} \geq 29$, so
    $u_0$ is light, and thus $u_1$ is also light. Since $\abs{u_0} > \tfrac12 \abs{w}$, the word $u_1$ is of type B. If
    $\abs{\alpha_1} \leq \abs{\beta_1}$, then $\Delta(u_0) = \Delta(u_1) \leq 3$ by \autoref{lem:split_bounds_1}, and
    this contradicts \autoref{lem:delta_lower_bound_6} (recall that we assume that $\abs{w} \geq 145$, so
    $\abs{u_0} > \tfrac12 \abs{w} \geq 72$). Therefore $\abs{\alpha_1} > \abs{\beta_1}$. Since $u_0 \alpha_1$ is a
    suffix of $w$, it follows that $\abs{u_0 \alpha_1} \leq \abs{w}$. Consequently, we have
    $\abs{\beta_1 u_2} \leq \abs{w}$ which means that $u_2$ is a factor of $\OL{w}$. This contradicts
    \autoref{lem:imbalance} since $u_0$ is a factor of $w$.

    The remaining cases are proved by applying the analogous \autoref{lem:split_bounds_2}.
  \end{proof}

  We next prove the main technical lemma of this part. The proof of \autoref{prp:odd_8_power} is almost immediate after
  this.

  \begin{lemma}\label{lem:8_greater_than_w}
    The word $\infw{F}$ does not contain abelian $8$-powers of period $m$ such that $m \leq \abs{w}$.
  \end{lemma}
  \begin{proof}
    Assume for a contradiction that $\infw{F}$ contains an abelian $8$-power $u_0 \dotsm u_7$ such that
    $\abs{u_0} \leq \abs{w}$. By \autoref{lem:odd_half}, we may assume that $\abs{u_0} > \tfrac12 \abs{w}$. By
    \autoref{lem:cube_1}, the words $u_0$, $\ldots$, $u_7$ are not factors of $w$ or $\OL{w}$. Therefore each $u_i$ is
    of type A or B. In fact, the words $u_0$, $u_2$, $u_4$, and $u_6$ are of the same type, as are $u_1$, $u_3$, $u_5$,
    and $u_7$. Moreover, the word $u_0$ is of type A if and only if $u_1$ is of type $B$.
 
   Notice that $v_0 v_1 v_2 v_3$, with $v_i = u_{2i} u_{2i+1}$, is an abelian $4$-power of period $2\abs{u_0}$
   occurring in $\infw{F}$. Let $M = \abs{f} - 2\abs{u_0}$. Since $\abs{u_0} \leq \abs{w} < \abs{f}/2$, we have
   $M > 0$. By applying \autoref{lem:reduction_short}, we see that $\infw{F}$ contains an abelian $4$-power
   $s_3 s_2 s_1 s_0$ of period $M$. In fact, by inspecting the proof of the aforementioned lemma, the abelian $4$-power
   $s_3 \dotsm s_0$ ends where $v_0 \dotsm v_3$ begins.
   
   Assume that $u_0$ is of type A, the other case being symmetric. Let us write $\OL{w} = \beta_{-1} \alpha_0$ for a
   word $\beta_{-1}$. Since $u_1$ is of type $B$, we may write $v_0 = \alpha_0 {\diamond} w \beta_1$. Moreover, we have
   $\beta_{-1} = \beta_1 s_0'$ with $\abs{s_0'} = M$. Since $s_3 \dotsm s_0$ ends where $v_0 \dotsm v_3$ begins, we see
   that $s_0' = s_0$. Repeating the argument for $v_i$, $i=1,2,3$, in place of $v_0$, we see that
   $\beta_{2i-1} = \beta_{2i + 1} s_i$. Hence $\beta_{-1} = \beta_7 s_3s_2s_1s_0$. But now $\OL{w}$ contains the
   abelian $4$-power $s_3 \dotsm s_0$, which is absurd.
  \end{proof}

  \begin{proof}[Proof of \autoref{prp:odd_8_power}]
    Suppose for a contradiction that $\infw{F}$ contains an abelian $8$-power of period $m$ such that
    $m < \abs{f} = 2\abs{w} + 1$. By \autoref{lem:reduction_short}, we may suppose that $m \leq \abs{w}$. However,
    \autoref{lem:8_greater_than_w} indicates that no such abelian power exists. This is a contradiction.
  \end{proof}

  \subsection{Even Length Case}\label{ssec:even_case}
  In this section, we prove the following two propositions.

  \begin{proposition}\label{prp:even_8_power_1}
    The word $g_1$ avoids abelian $8$-powers cyclically if $\abs{w}$ is odd.
  \end{proposition}

  \begin{proposition}\label{prp:even_8_power_2}
    The word $g_2$ avoids abelian $8$-powers cyclically if $\abs{w}$ is even.
  \end{proposition}

  As the reader might have observed, the letter $\diamond$ often did not play a particular role in the proofs of
  \autoref{ssec:odd_case}. This means that the previous lemmas transfer to the case of the words $\infw{G}_1$ and
  $\infw{G}_2$ mostly intact. Consequently, we omit repetitive details from the proofs of this section and indicate
  only what has changed.

  Similar to \autoref{ssec:odd_case}, let $u_0 \dotsm u_{N-1}$ be an abelian $N$-power occurring in $\infw{G}_1$ such
  that $\abs{u_0} \leq \abs{w}$, and consider a word $u_i$ for some $i$. We classify the word $u_i$ as follows.
  \begin{enumerate}[(A)]
    \item $u_i = \alpha_i \beta_i$ for a nonempty suffix $\alpha_i$ of $\OL{w}$ and a nonempty prefix $\beta_i$ of $w$.
    \item $u_i = \alpha_i \beta_i$ for a nonempty suffix $\alpha_i$ of $w$ and a nonempty prefix $\beta_i$ of $\OL{w}$.
  \end{enumerate}
  For an abelian $N$-power $u_0 \dotsm u_{N-1}$ occurring in $\infw{G}_2$ such that $\abs{u_0} \leq \abs{w}$, we
  define the type of $u_i$ as follows.
  \begin{enumerate}[(A)]
    \item $u_i = \alpha_i \beta_i$ for a nonempty suffix $\alpha_i$ of $\OL{w}^\bullet$ and a nonempty prefix $\beta_i$ of $w$.
    \item $u_i = \alpha_i \beta_i$ for a nonempty suffix $\alpha_i$ of $w$ and a nonempty prefix $\beta_i$ of $\OL{w}^\bullet$.
  \end{enumerate}

  Propositions \ref{prp:even_8_power_1} and \ref{prp:even_8_power_2} can be again verified when $\abs{w} < 145$, so we
  assume that $w$ has length at least $145$ for the remainder of this section. In the following lemmas, we shall make
  no use of the parity of $\abs{g_1}$ or $\abs{g_2}$. In fact, the parity shall only play a role in the proofs of
  \autoref{prp:even_8_power_1} and \autoref{prp:even_8_power_2} at the end of this section.

  \begin{lemma}\label{lem:even_half}
    If an abelian $8$-power of period $m$ occurs in $\infw{G}_1$ or $\infw{G}_2$, then $m > \tfrac12 \abs{w}$.
  \end{lemma}
  \begin{proof}
    Assume for a contradiction that either of the words contains an abelian $8$-power $u_0 \dotsm u_7$ with period
    $m \leq \tfrac{1}{2} \abs{w}$. We first show that $u_i$ is of type A or type B for some $i$. Assume the contrary
    that no $u_i$ is of type A or type B. Say the word $u_0$ occurs in $w$ and that $w$ is followed by $w'$ where
    $w' \in \{\OL{w}, \OL{w}^\bullet\}$. Since $w$ avoids abelian $4$-powers, one of the words $u_1$, $u_2$, or $u_3$,
    say $u_j$, is a prefix of $w'$ (since they do not have a type). Since $j \leq 3$, we see that $u_{j+4}$ exists.
    There exists $u_k$ such that $u_k$ is a prefix of $w$ and $w' = u_j u_{j+1} \dotsm u_{k-1}$ for otherwise the
    abelian $5$-power $u_j u_{j+1} \dotsm u_{j+4}$ is a prefix of $w'$, but neither $\OL{w}$ nor $\OL{w}^\bullet$ can
    have such a factor. It follows that either $w'$ is an abelian $N$-power for some $N \leq 4$ or $w' = u_j$. In the
    former case, we have $\abs{u_0} = \abs{w'}/N \geq 145/4 \geq 36$, so $u_0$ is light by \autoref{lem:light}.
    However, the word $u_j$, a proper prefix of $w'$, is heavy by \autoref{lem:light}. Therefore it must be that
    $w' = u_j$, but this contradicts the assumption that $\smash[t]{\abs{u_j} \leq \tfrac12 \abs{w}}$. The case that
    $u_0$ occurs in $w'$ is symmetric.
    
    To conclude the proof, we may now follow the proof of \autoref{lem:odd_half}. Notice in particular that if $u_i$ is
    of type A, then $\alpha_i$ is nonempty, and thus the change of the final letter of $\OL{w}$ does not affect
    $u_{i-1}$. 
  \end{proof}

  The following two lemmas are combinations of Lemmas \ref{lem:split_bounds_1} and \ref{lem:split_bounds_2} adjusted
  for the words $\infw{G}_1$ and $\infw{G}_2$.

  \begin{lemma}\label{lem:even_split_bounds_1}
    Suppose that the word $\infw{G}_1$ contains an abelian $N$-power $u_0 \dotsm u_{N-1}$. Suppose that $u_i$ is of
    type A.
    \begin{enumerate}[(i)]
      \item If $\abs{\alpha_i} \geq \abs{\beta_i}$, then $\Delta(u_i) \leq 3$.
      \item If $\abs{\alpha_i} \leq \abs{\beta_i}$, then $\Delta(u_i) \geq -3$.
    \end{enumerate}
    Suppose that $u_i$ is of type B.
    \begin{enumerate}[(i)]
      \item If $\abs{\alpha_i} \geq \abs{\beta_i}$, then $\Delta(u_i) \geq -3$.
      \item If $\abs{\alpha_i} \leq \abs{\beta_i}$, then $\Delta(u_i) \leq 3$.
    \end{enumerate}
  \end{lemma}
  \begin{proof}
    Follow the proof of \autoref{lem:split_bounds_1}.
  \end{proof}

  \begin{lemma}\label{lem:even_split_bounds_2}
    Suppose that the word $\infw{G}_2$ contains an abelian $N$-power $u_0 \dotsm u_{N-1}$. Suppose that $u_i$ is of
    type A.
    \begin{enumerate}[(i)]
      \item If $\abs{\alpha_i} \geq \abs{\beta_i}$, then $\Delta(u_i) \leq 5$.
      \item If $\abs{\alpha_i} \leq \abs{\beta_i}$, then $\Delta(u_i) \geq -1$.
    \end{enumerate}
    Suppose that $u_i$ is of type B.
    \begin{enumerate}[(i)]
      \item If $\abs{\alpha_i} \geq \abs{\beta_i}$, then $\Delta(u_i) \geq -3$.
      \item If $\abs{\alpha_i} \leq \abs{\beta_i}$, then $\Delta(u_i) \leq 3$.
    \end{enumerate}
  \end{lemma}
  \begin{proof}
    We show how to handle the cases (i). Say $u_i$ is of type A and $\abs{\alpha_i} \geq \abs{\beta_i}$. We may write
    $\alpha_i = z \OL{\beta_i}^\bullet$ for a word $z$ (here we have $\abs{\beta_i} > 0$ by definition). It follows
    that $\Delta(u_i) = \Delta(z) + 2$. Since $z$ is a factor of $\OL{w}$, we have $\Delta(z) \leq 3$, so
    $\Delta(u_i) \leq 5$.

    Suppose that $u_i$ is of type B and $\abs{\alpha_i} \geq \abs{\beta_i}$. Since $\abs{\alpha_i} > 0$ by its
    definition and $\abs{u_i} \leq \abs{w}$, we see that $\abs{\beta_i} < \abs{w}$. It follows that
    $\alpha_i = z \OL{\beta_i}$ for a word $z$. Thus $\Delta(u_i) = \Delta(z)$. Since $z$ is a factor of $w$, we see
    that $\Delta(z) \geq -3$.
  \end{proof}

  \begin{lemma}\label{lem:even_cube_1}
    Let $u_0 u_1 u_2$ be an abelian $3$-power occurring in $\infw{G}_1$. If
    \begin{enumerate}[(i)]
      \item $u_0$ occurs in $w$ or $\OL{w}$ or
      \item $u_2$ occurs in $w$ or $\OL{w}$,
    \end{enumerate}
    then $\abs{u_0} \leq \tfrac12 \abs{w}$.
  \end{lemma}
  \begin{proof}
    Follow the proof of \autoref{lem:cube_1} and apply \autoref{lem:even_split_bounds_1} appropriately.
  \end{proof}

  \begin{lemma}\label{lem:even_cube_2}
    Let $u_0 u_1 u_2$ be an abelian $3$-power occurring in $\infw{G}_2$. If
    \begin{enumerate}[(i)]
      \item $u_0$ occurs in $w$ or $\OL{w}^\bullet$ or
      \item $u_2$ occurs in $w$ or $\OL{w}^\bullet$,
    \end{enumerate}
    then $\abs{u_0} \leq \tfrac12 \abs{w}$.
  \end{lemma}
  \begin{proof}
    We show how to handle the case where $u_0$ occurs in $\OL{w}^\bullet$. Assume on the contrary that
    $\abs{u_0} > \tfrac12 \abs{w}$ and $u_0$ occurs in $\OL{w}^\bullet$. Suppose first that $u_0$ is a suffix of
    $\OL{w}^\bullet$. If $\abs{u_0} = \abs{w}$, then $u_1 = w$ and consequently $\OL{w}^\bullet$ and $w$ are abelian
    equivalent. This means that $\abs{w}_0 = \abs{\OL{w}}_1 = \abs{\OL{w}^\bullet}_1 + 1 = \abs{w}_1 + 1$. Thus
    $\Delta(w) = -1$, and this contradicts \autoref{lem:light}. Therefore $\abs{u_0} < \abs{w}$ implying that $u_1$ is
    a factor of $w$. Thus $\Delta(u_1) \geq 6$ by \autoref{lem:delta_lower_bound_6}. On the other hand, by taking into
    account the changed final letter of $\OL{w}^\bullet$, \autoref{lem:difference_upper_bound} implies that
    $\Delta(u_0) \leq -1$, so it is not possible that $\Delta(u_0) = \Delta(u_1)$.

    We may thus assume that $u_0$ is not a suffix of $\OL{w}^\bullet$. Since $\abs{u_0} > \tfrac12 \abs{w}$, it follows
    that the word $u_1$ is of type A. If $\abs{\alpha_1} \leq \abs{\beta_1}$, then $\Delta(u_0) = \Delta(u_1) \geq -1$
    by \autoref{lem:even_split_bounds_2}. This contradicts \autoref{lem:delta_lower_bound_6}, so
    $\abs{\alpha_1} > \abs{\beta_1}$. Since $u_0 \alpha_1$ is a suffix of $\OL{w}^\bullet$, it follows that
    $\abs{u_0 \alpha_1} \leq \abs{w}$. Hence $\abs{\beta_1 u_2} \leq \abs{w}$ and $u_2$ is a factor of $w$. This
    contradicts \autoref{lem:imbalance}.
  \end{proof}

  \begin{lemma}\label{lem:8_less_than_w_even_1}
    The word $\infw{G}_1$ does not contain abelian $8$-powers of period $m$ such that $m < \abs{w}$.
  \end{lemma}
  \begin{proof}
    The proof of \autoref{lem:8_greater_than_w} works mostly as it is for the word $\infw{G}_1$. Indeed, Lemmas
    \ref{lem:even_half} and \ref{lem:even_cube_1} guarantee that each $u_i$ is of type A or B. Set
    $M = \abs{g_1} - 2\abs{u_0}$. Notice that we assume $\abs{u_0} < \abs{w} = \abs{g_1}/2$, so $M > 0$. The remaining
    arguments are the same, only the $\diamond$ symbol is omitted. The conclusion is that $\OL{w}$ contains an abelian
    $4$-power of period $M$ ending at position $\abs{w} - \abs{\alpha_0}$. This is impossible.
  \end{proof}

  \begin{lemma}\label{lem:8_less_than_w_even_2}
    The word $\infw{G}_2$ does not contain abelian $8$-powers of period $m$ such that $m < \abs{w}$.
  \end{lemma}
  \begin{proof}
    We proceed as in the proof of \autoref{lem:8_greater_than_w}. By Lemmas \ref{lem:even_half} and
    \ref{lem:even_cube_2}, we may suppose that each $u_i$ is of type A or B. Set $M = \abs{g_2} - 2\abs{u_0}$. Again,
    $\abs{u_0} < \abs{w}$ is assumed so $M > 0$. Following the arguments of \autoref{lem:8_greater_than_w} (omitting
    $\diamond$), we find that $\OL{w}^{\bullet}$ contains an abelian $4$-power ending at position
    $\abs{w} - \abs{\alpha_0}$. Observe that since the repetition is not a suffix of $\OL{w}^{\bullet}$ (as
    $\abs{\alpha_0} > 0$), the same abelian $4$-power occurs in $\OL{w}$. This is a contradiction.
  \end{proof}

  \begin{proof}[Proof of \autoref{prp:even_8_power_1}]
    Say $\abs{w}$ is odd, and suppose for a contradiction that $\infw{G}_1$ contains an abelian $8$-power of period $m$
    such that $m < \abs{g_1} = 2\abs{w}$. By \autoref{lem:reduction_short}, we may suppose that $m \leq \abs{w}$.
    \autoref{lem:8_less_than_w_even_1} implies that $m = \abs{w}$. By \autoref{lem:even_cube_1}, it is not possible
    that $u_i = w$ or $u_i = \OL{w}$ for some $i$. Therefore all $u_i$ are of type A or B. We handle the case that
    $\abs{\alpha_0} \geq \abs{\beta_0}$; the case $\abs{\alpha_0} \leq \abs{\beta_0}$ is symmetric. Write
    $\alpha_0 = z_0 \OL{\beta_0}$ so that $\Delta(u_0) = \Delta(z_0)$. When $u_0$ is of type A, the word $w$ has prefix
    $\beta_0 \OL{z_0}$ and suffix $z_1 \OL{\beta_1}$ (here $u_1 = \alpha_1 \beta_1 = z_1 \OL{\beta_1} \beta_1$). Since
    $m = \abs{w}$, we have $\abs{\beta_0} = \abs{\beta_1}$. Since
    $m = 2\abs{\beta_0} + \abs{z_0} = 2\abs{\beta_1} + \abs{z_1}$, we conclude that $\OL{z_0} = z_1$. The same
    conclusion is reached if $u_0$ is type B. Since $\Delta(u_0) = \Delta(u_1) = \Delta(z_1)$, we have
    $\Delta(z_0) = \Delta(z_1) = \Delta(\OL{z_0})$, so $\Delta(z_0) = 0$. Therefore $\abs{z_0}$ is even. Since
    $\abs{w} = m = 2\abs{\beta_0} + \abs{z_0}$, it follows that $\abs{w}$ is even. This is contrary to our hypothesis
    that $\abs{w}$ is odd.
  \end{proof}

  \begin{proof}[Proof of \autoref{prp:even_8_power_2}]
    Suppose that $\abs{w}$ is even, and assume for a contradiction that $\infw{G}_2$ contains an abelian $8$-power of
    period $m$ with $m < 2\abs{w}$. As in the proof of \autoref{prp:even_8_power_1}, we see that it must be that
    $m = \abs{w}$. Moreover, the words $u_i$ are of type A or B by \autoref{lem:even_cube_2}. Suppose that $u_0$ is of
    type A and $\abs{\alpha_0} \geq \abs{\beta_0}$. The remaining cases are similar. Write
    $u_0 = \alpha_0^\bullet \beta_0 = z_0 \OL{\beta_0}^\bullet \! \beta_0$ and
    $u_1 = \alpha_1 \beta_1 = z_1 \OL{\beta_1} \beta_1$ for some words $z_0$ and $z_1$ of the same length. Therefore
    $\Delta(u_0) = \Delta(z_0) + 2 = \Delta(u_1) = \Delta(z_1)$. Since $\abs{\beta_0} > 0$, it is straightforward to
    see that $\OL{z_0} = z_1$. Thus $\Delta(z_0) + 2 = \Delta(\OL{z_0})$, that is,
    $\abs{z_0}_0 - \abs{z_0}_1 + 2 = \abs{z_0}_1 - \abs{z_0}_0$. It follows that $\abs{z_0}_0 + 1 = \abs{z_0}_1$, and
    so $\abs{z_0} = \abs{z_0}_0 + \abs{z_0}_1 = 2\abs{z_0}_0 + 1$. Therefore $\abs{z_0}$ is odd, and consequently
    $\abs{w} = 2\abs{\beta_0} + \abs{z_0}$ is odd. This is a contradiction.
  \end{proof}

  Propositions \ref{prp:even_8_power_1} and \ref{prp:even_8_power_2} together with \autoref{prp:odd_8_power} imply
  \autoref{thm:binary_case}.

  \section{Avoiding Ordinary Powers Cyclically}\label{sec:ordinary_powers}
  As mentioned in the introduction, previous research has considered the avoidance of ordinary powers in circular
  words. A circular word is simply a conjugacy class of words, that is, a word $w$ avoids $N$-powers circularly if none
  of the conjugates of $w$ contains an $N$-power as a factor. This constrains the periods to have length at most
  $\floor{\abs{w}/N}$ while our definition of cyclic avoidance disallows periods up to length $\abs{w} - 1$. The
  purpose of this section is to generalize the known results on circular avoidance of powers to our cyclic setting.

  \begin{definition}
    Let $w$ be a word. Then $w$ avoids $N$-powers \emph{cyclically} if for each $N$-power of period $m$ occurring in
    the infinite word $w^\omega$, we have $m \geq \abs{w}$.
  \end{definition}

  The following analogue of \autoref{lem:reduction_short} is straightforward to prove.

  \begin{lemma}\label{lem:ordinaryShortPowers}
    Assume that $x^\omega$ contains an $N$-power of period $m$ with $\frac12 \abs{x} \leq m < \abs{x}$. Then it
    contains an $N$-power with period $\abs{x} - m$.
  \end{lemma}

  This lemma implies that the concepts of avoiding $2$-powers circularly and avoiding $2$-powers cyclically are the
  same concept. When $N > 2$, this is not true. For example, the word $00$ avoids abelian $N$-powers circularly for
  $N > 2$, but it never avoids abelian $N$-powers cyclically.

  Currie proved in \cite{2002:there_are_ternary_circular_square-free} that if $n \notin \{5, 7, 9, 10, 14, 17\}$ then
  there exists a word of length $n$ over a $3$-letter alphabet that avoids $2$-powers circularly. By the preceding
  paragraph, we have the following result (notice that $2$-powers cannot be avoided with just two letters).

  \begin{theorem}\label{thm:ordinary_ternary}
    If $n \notin \{5, 7, 9, 10, 14, 17\}$ then there exists a word of length $n$ over a $3$-letter alphabet that avoids
    $2$-powers cyclically.
  \end{theorem}

  Notice that for $n \in \{5, 7, 9, 10, 14, 17\}$ there exists a word of length $n$ over a $4$-letter alphabet avoiding
  $2$-powers cyclically. Such words are, e.g., $01023$, $0102013$, $010201203$, $0102010313$, $01020103010213$, and
  $01020103010212313$. Notice in addition that for each $n$ there exists a word of length $n$ over a $3$-letter
  alphabet that avoids $2^+$-powers cyclically (see below for the definition). To see this, it is sufficient to observe
  that the words $00102$, $0010012$, $001001102$, $0010011202$, $00100112001002$, and $00100112001001202$ avoid
  $2^+$-powers cyclically.

  What is left is to determine the least exponent $N$ such that for all $n$ there exists a binary word $w$ of length
  $n$ such that $w$ avoids $N$-powers cyclically. In the context of ordinary powers, it is natural to consider
  fractional exponents, and thus we give the following definition. We do not consider fractional abelian exponents in
  this paper; for discussion on this concept, see
  \cite{1999:words_strongly_avoiding_fractional_powers,2012:on_abelian_repetition_threshold}.

  \begin{definition}
    Let $w$ be a word and $N$ be a rational number such that $N > 1$. Then $w$ avoids $N^+$-powers \emph{cyclically} if
    for each $N^+$-power of period $m$ occurring in the infinite word $w^\omega$, we have $m \geq \abs{w}$. A word $u$
    is an $N^+$-power if $u$ is an $R$-power for some $R > N$.
  \end{definition}

  Let $\infw{t}$ be the fixed point $\sigma^\omega(0)$ of the morphism $\sigma\colon 0 \mapsto 01$, $1 \mapsto 10$. The
  word $\infw{t}$ is the famous Thue-Morse word; see \cite[Sect.~1.6]{2003:automatic_sequences}. Aberkane and Currie
  proved in \cite{2005:the_thue-morse_word_contains_circular} that the Thue-Morse word $\infw{t}$ contains a factor
  avoiding $5/2^+$-powers circularly for all lengths. We generalize this result to our cyclic setting. This result
  implies \autoref{thm:ordinary_binary}.

  \begin{theorem}\label{thm:thue-morse}
    For each $n$, there exists a factor of length $n$ of the Thue-Morse word avoiding $5/2^+$-powers cyclically.
  \end{theorem}

  It can be shown that the exponent $5/2$ is optimal for binary words by inspecting all binary words of length $5$.

  In order to prove \autoref{thm:thue-morse}, we employ the automatic theorem-proving software Walnut \cite{walnut}.
  Properties of automatic sequences \cite{2003:automatic_sequences} that are expressible in a certain first-order logic
  are decidable, and Walnut implements the decision procedure. The Thue-Morse word $\infw{t}$ is a $2$-automatic word,
  so Walnut is applicable. We wish to keep the discussion on the decision procedure and usage of Walnut brief, so we
  merely describe the logical formulas necessary to encode our problem and refer the reader to
  \cite{2019:circularly_squarefree_words_and_unbordered_conjugates} for a proof of \autoref{thm:ordinary_ternary} using
  Walnut. See also \cite{2019:circular_critical_exponents_for_thue-morse_factors}.

  Let $w$ be a factor of the Thue-Morse word. If $w^\omega$ contains an $N$-power of period $m$ such that $m < \abs{w}$
  and $N > 3$, then $w^\omega$ contains a $3$-power of period $m$. Therefore in order to show that $w$ avoids
  $5/2^+$-powers cyclically, we only need to consider $N$-powers with $5/2 < N \leq 3$. Notice that such a power $u$ is
  necessarily a factor of $w^4$. We first write a predicate $\operatorname{cRepK}(i,j,m,n,p)$, $K = 1, \ldots, 4$, that
  evaluates to true if and only if the factor $w$ of length $n$ beginning at the position $i$ of the Thue-Morse word
  $\infw{t}$ is such that $w^K$ has a factor $u$ of length $m$ beginning at position $j$, $i \leq j < i + n$, such that
  $u$ has period $p$ and $u$ is not a factor of $w^{K-1}$. The predicate needs to be written somewhat awkwardly as
  $w^K$ is not necessarily a factor of $\infw{t}$. For example, we have
  \begin{align*}
    \operatorname{cRep2}(i,j,m,n,p) &= (i \leq j < i+n) \land (i+n \leq j+m \leq i+2n) \land \\
    & \hspace{1.4em} (\forall k (j \leq k < i+n-p) \Longrightarrow \infw{t}[k] = \infw{t}[k+p]) \land \\
    & \hspace{1.4em} (\forall k (i+n-p \leq k < i+n) \Longrightarrow \infw{t}[k] = \infw{t}[k+p-n]) \land \\
    & \hspace{1.4em} (\forall k (i+n \leq k < j+m-p) \Longrightarrow \infw{t}[k-n] = \infw{t}[k+p-n]).
  \end{align*}
  We can then write a predicate $\operatorname{ncyc}(i,n)$ that evaluates to true if and only if the factor $w$ of
  $\infw{t}$ of length $n$ starting at position $i$ is such that $w^4$ contains a factor that has period $p$ with
  $5/2 < \abs{u}/p \leq 3$. Its definition is:
  \begin{align*}
    \operatorname{ncyc}(i,n) &= \exists j,m,p ((0 < p < n) \land (5p < 2m \leq 6p) \land \\
    & \hspace{1.4em} (cRep1(i,j,m,n,p) \lor cRep2(i,j,m,n,p) \lor cRep3(i,j,m,n,p) \lor  \\
    & \hspace{1.4em} cRep4(i,j,m,n,p))).
  \end{align*}
  Finally the following predicate evaluates to true if and only if \autoref{thm:thue-morse} is true:
  \begin{align*}
    \forall n ( (n > 0) \Longrightarrow (\exists i \lnot \operatorname{ncyc}(i, n))).
  \end{align*}
  Inputting the above predicates to Walnut produces an automaton accepting all inputs meaning that
  \autoref{thm:thue-morse} is true.

  \section{Discussion on Future Research}\label{sec:conclusions}
  Obviously the main question is what is the value of $\aci{k}$ for $k = 2, 3, 4$. \autoref{thm:main2} seems to support
  the claims that $\aci{2} = 4$, $\aci{3} = 3$, and $\aci{4} = 2$, but the first claim is false as there is no binary
  word of length $8$ avoiding abelian $4$-powers cyclically. This leads us to ask the following questions.

  \begin{question}
    Is it the case that $\aci{2} = 5$, $\aci{3} = 3$, and $\aci{4} = 2$?
  \end{question}

  \begin{question}
    If $n \neq 8$, does there exist a word of length $n$ over a $2$-letter alphabet avoiding abelian $4$-powers
    cyclically?
  \end{question}

  Our computer experiments have not found a counterexample to the above questions among lengths less than $150$. Notice
  that our question whether $\aci{4} = 2$ is stronger than the conjecture of \cite{2018:anagram-free_graph_colouring}
  mentioned in the preliminaries after \autoref{thm:main2}. A positive answer to the latter question would imply that
  $\aci{2} = 5$ as the word $00001011$ of length $8$ avoids abelian $5$-powers cyclically.

  We do not know how to approach these questions. The lowest hanging fruit is to improve the construction of
  \autoref{sec:bounds_for_aci} and lower the upper bound on $\aci{2}$. We remark that the particular construction given
  here cannot be used to improve the upper bound $8$ in \autoref{thm:binary_case} as some of the words constructed
  contain abelian $7$-powers with short period. If two words that avoid abelian $4$-powers are concatenated, then a
  priori abelian $7$-powers could appear. An improved construction would need to take special care to concatenate the
  words in such a way that their respective abelian $3$-powers of common period do not appear too close to each other.
  It seems that no precise information on the structure and location of abelian $3$-powers in words that avoid abelian
  $4$-powers is found in the literature. Even the sets of possible periods of abelian powers occurring in infinite
  words have been studied very little. The only papers in this direction are the papers
  \cite{2016:abelian_powers_and_repetitions_in_sturmian_words,2020:abelian_periods_of_factors_of_sturmian_words}
  concerning the abelian period sets of Sturmian words. This knowledge however is not helpful in this context as
  Sturmian words contain abelian powers of arbitrarily high exponent
  \cite[Proposition~4.10]{2016:abelian_powers_and_repetitions_in_sturmian_words}. It seems that making such
  concatenation arguments work for the alphabet sizes $3$ and $4$ is even more difficult especially because there is
  less room for improvement.

  An alternative way to improve our results would be to find infinite words whose language contains the sought words.
  For example, Justin's morphism $0 \mapsto 00001$, $1 \mapsto 01111$ seems promising
  \cite{1973:characterization_of_the_repetitive_commutative_semigroups}. It has a factor of length $n$ avoiding abelian
  $5$-powers cyclically for $n = 1, \ldots, 400$. We do not know how to prove that such a factor exists for each
  length. Since the the fixed point of Justin's morphism is automatic, it might be possible to attack this problem via
  automatic theorem-proving as in \autoref{sec:ordinary_powers}. The problem in this plan is that this type of
  automatic theorem-proving requires the problem to be written in a certain restricted first order logic and generally
  abelian properties of words cannot be expressed in this logic \cite[Sect.~5.2]{mthesis:luke}.

  We have dealt in this paper only with the question of existence. A significantly harder problem would be to provide a
  lower bound, for example, for the number of binary words of length $n$ that avoid abelian $4$-powers cyclically. We
  have recorded this sequence as the sequence \href{https://oeis.org/A334831}{A334831} in Sloane's \emph{On-Line
  Encyclopedia of Integer Sequences} \cite{oeis}. The first values of the sequence are $2$, $2$, $6$, $8$, $10$, $6$,
  $28$, $0$, $36$, $120$, $132$, $168$, $364$, $112$.
  
  \section*{Acknowledgements}
  We thank the reviewer for remarks that improved the quality of the paper. We also thank him/her for pointing out the
  conjecture in \cite{2018:anagram-free_graph_colouring}.

  \printbibliography

\end{document}